\documentclass[11pt]{scrartcl}
\usepackage[utf8]{inputenc}
\usepackage{amsfonts,amsmath,amsthm}
\usepackage{graphicx}
\usepackage{enumerate}
\usepackage{trsym}
\usepackage{here}

\newtheorem{theorem}{Theorem}
\newtheorem{corollary}[theorem]{Corollary}
\newtheorem{definition}[theorem]{Definition}
\newtheorem{example}[theorem]{Example}
\newtheorem{lemma}[theorem]{Lemma}
\newtheorem{proposition}[theorem]{Proposition}
\newtheorem{remark}[theorem]{Remark}

\newcommand{\lasttrade}{\Upsilon}
\newcommand{\nexttrade}{\Xi}

\newcommand{\nbid}{\underline{n}}
\newcommand{\nask}{\overline{n}}
\newcommand{\ebid}{\underline{e}}

\DeclareMathOperator{\BES}{BES}
\DeclareMathOperator{\csch}{csch}

\DeclareMathOperator{\sgn}{sgn}
\input{tcilatex}
\begin{document}

\title{Brownian trading excursions and avalanches}
\author{Friedrich Hubalek\thanks{%
Financial and Actuarial Mathematics, Vienna University of Technology,
Wiedner Hauptstraße~8/ 105-1, 1040 Vienna, Austria. phone
+43-1-58801-10511 \ fax +43-1-58801-9-105199 (\texttt{%
fhubalek@fam.tuwien.ac.at})} \and Paul Krühner\thanks{%
Financial and Actuarial Mathematics, Vienna University of Technology,
Wiedner Hauptstraße~8/ 105-1, 1040 Vienna, Austria. phone
+43-1-58801-10552 \ fax +43-1-58801-9-10552 (\texttt{%
paulkrue@fam.tuwien.ac.at})} \and Thorsten Rheinländer\thanks{%
Financial and Actuarial Mathematics, Vienna University of Technology,
Wiedner Hauptstraße~8/ 105-1, 1040 Vienna, Austria. phone
+43-1-58801-10550 \ fax +43-1-58801-9-10550 (\texttt{%
rheinlan@fam.tuwien.ac.at})} }
\maketitle

\begin{abstract}
{We study a parsimonious but non-trivial model of the latent limit order book
where orders get placed with a fixed displacement from a center price process,
i.e.\ some process in-between best bid and best ask, and get
executed whenever this center price reaches their level. This mechanism
corresponds to the fundamental solution of the stochastic heat equation with
multiplicative noise for the relative order volume distribution. We classify
various types of trades, and introduce the trading excursion process which is
a Poisson point process. This allows to derive the Laplace transforms of the
times to various trading events under the corresponding intensity measure. As
a main application, we study the distribution of order avalanches, i.e.\ a
series of order executions not interrupted by more than an $\varepsilon$-time
interval, which moreover generalizes recent results about Parisian options.}

\end{abstract}

\section{Introduction}

The main object of interest in this study is to develop a parsimonious model
of the limit order book (LOB) for financial assets, where price level and
number of orders away from the best bid/ask prices are recorded. We refer to
\cite{Cartea} for an overview of market
microstructure trading.

Quite a few articles on the LOB, starting amongst others with Kruk \cite{K},
are investigating the limiting behavior of some discretely modeled dynamics.
Cont and de Larrard \cite{CL} model the dynamics of best bid and ask quotes as
two interacting queues. Their structural model combines high frequency price
dynamics with the order flow, and a Markovian jump-diffusion process in the
positive orthant is reached as scaling limit. Horst et al.\ \cite{BHQ} derive
a functional limit theorem where the limits of the standing buy and sell
volume densities are described by two linear stochastic partial differential
equations, which are coupled with a two-dimensional reflected Brownian motion
that is the limit of the best bid and ask price processes, whereas Abergel and
Jedidi \cite{AJ} consider the volume of the LOB at different distances to the
best ask price and determine a diffusion limit for the mid price. Delattre et
al.\ \cite{dLRR} study the efficient price which is a price market
practitioners could agree upon and its statistical estimation. The placing of
orders is captured by Osterrieder \cite{O} in a marked point process model, so
that the order book is modeled by several measure valued processes.

Our study is quite different to the aforementioned works. For the point of
focus, we consider a \emph{latent order book model}, see \cite{TothEtAl2011},
which contains the orders of low-frequency traders, whereas high
frequency orders which get typically cancelled after a very short time span
are not recorded. As we are in particular interested how limit orders get
intrinsically executed, we do not allow for any other mechanism besides that
the center price, which we model as a Brownian motion, hits the level where the
limit orders are placed. Here orders get issued relative to the actual center
price according to some universal aggregated volume density function.

Expanding formally the relative order volume distribution via Ito's formula,
it results that this volume distribution solves a stochastic heat equation
with multiplicative noise, which will be studied in a subsequent paper
\cite{HKR-SPDE}. Here we are interested in
the fundamental solution, which corresponds to order placement according to a
Dirac measure on some level $\mu$ away from the best bid or ask price. This
leads to an approachable, but nonetheless highly non-trivial model of limit
order executions. We do not make any claims that our model is realistic, but
it should be understood as a parsimonious model which can later be extended in
various directions, like more sophisticated models for the order arrival
process as well as for the center price.

In this context, we discuss in detail and classify various types of trading
times which can be characterized via doubly reflected Brownian motion. {There
are two basic execution mechanisms for the ask side of the book (which one can
then subdivide further): a Type~{I} trade occurs whenever the price maximum
increases, whereas a Type~{II} trade is triggered after a downfall by more
than the displacement followed by an equal surge of the center price. }We then
study excursions to the next trading time. The trading excursion process is a
Poisson point process, for which the intensity measure is known. This allows
us to calculate the Laplace transforms under the intensity measure of the
times to various types of trades in terms of hypberbolic functions.

A major application of these results is the study of order execution
avalanches, {i.e.\ a series of order executions not interrupted by more than
an $\varepsilon$-time interval.} One has to allow for a small time window
where orders do not get executed due to the fact that Brownian motion has no
point of increase. Here we drew some inspiration from the paper Stapleton and
Christensen \cite{SC} about avalanches which is in the spirit of the theory of
self-organized criticality. We derive the Laplace transform of the general
avalanche length of order execution in our model, which improves over several
known results in the context of Parisian options before, in particular by
Dassios and Wu \cite{DW} as well as Gauthier \cite{G}. A similar result for
simple avalanches (not containing Type II trades) has been proved by Dudok de
Wit \cite{DDW} by a different method in a limit order book framework. 

The structure of the paper is as follows: In the next section, we introduce
our latent limit order book model, in particular the order placement and
execution mechanisms. Section 3 contains the classification of various types
of trading times, followed by an analysis of the order book with Dirac-type
placement. In Section 5 the central idea of trading excursions is introduced,
which leads in Section 6 to the hyperbolic function table regarding Laplace transforms
of the times to various trading events. As our main application, we derive in
Section 7 the Laplace transform of the order avalanche length.

\section{A Brownian motion model for the limit order book}

\label{sec:brown} As in \cite[Sec.XII.2, p.480]{RY} we shall work with the
canonical version $W$ of Brownian motion on the Wiener space $(\mathbf{W},%
\mathcal{F},P)$. This means $\mathbf{W}$ is the space of continuous
functions $w:\mathbb{R}_+\to\mathbb{R}$ with $w(0)=0$, equipped with the
locally uniform topology, $P$ is the Wiener measure, $\mathcal{F}$ is the
Borel $\sigma$-field of $\mathbf{W}$ completed with respect to $P$, and $%
W_t(w)=w(t)$ for $t\geq0$ and $w\in\mathbf{W}$.

We denote by $\{L_{t}^{a}:a\in \mathbb{R},t\in \mathbb{R}_{+}\}$ a
bicontinuous modification of the family of local times of~$W$, see \cite[%
Thm.VI.1.7, p.225]{RY}.

We assume that orders arrive with density one in every infinitesimal time
interval $dt$, model the center price process $W$ as a Brownian motion, and
denote by $\mathbb{F=}\left( \mathcal{F}_{t}\right) _{t\geq 0}$ its
augmented filtration. This `center price' is just thought to lie in between
best bid and ask, see below for the precise order execution mechanism at
best bid/ask.

\subsection{Absolutely continuous order placement}

\label{s:ac order placement} Let $V(t,x)$ denote the order volume at time $%
t\geq0$ and level $x\in\mathbb{R}$. The placement of new limit orders is
governed by some integrable function $g:\mathbb{R}\setminus\{0\}\to\mathbb{R}%
_+$. An intuitive description of the dynamics is as follows:

\begin{itemize}
\item During an infinitesimal time interval $dt$, it is assumed that new
limit orders are created at every level $W_t+x$ with volume density $g(x)dx$,

\item limit orders at level $x$ are executed once the center price hits the
corresponding level, i.e.\ when $W_t=x$,

\item there will be no order withdrawal.
\end{itemize}

For a rigorous definition let us denote by 
\begin{equation}
\sigma (t,x):=\sup \{s\in \lbrack 0,t]:\mbox{$W_s=x$ or $s=0$}\}
\end{equation}%
the \emph{last exit time of} $W$ from level $x$ before time $t$. Consider
now a time $t\geq 0$ and a level $x\in \mathbb{R}$. At time $\sigma (t,x)$
all orders at level $x$ are executed. The volume $V(t,x)$ is made up from
new orders placed during the time interval $(\sigma (t,x),t]$ and thus we
define 
\begin{equation}
V(t,x):=\int_{\sigma (t,x)}^{t}g(x-W_{s})\text{ }ds.  \label{volume function}
\end{equation}%
Note that order execution is included in (\ref{volume function}) since the
integral gets void once $W$ reaches the level $x$, capturing the
aforementioned execution mechanism. In particular $V(t,W_{t})=0$ for all $%
t\geq 0$.

We distinguish the bid order book and the ask order book processes, which we
define as 
\begin{equation}
\underline{V}(t,x):=V(t,x)I_{x\leq W_{t}},\quad \overline{V}%
(t,x):=V(t,x)I_{x\geq W_{t}}.  \label{bid-ask}
\end{equation}%
Obviously\footnote{%
We have $\underline{V}(t,x)\geq 0$ and $\overline{V}(t,x)\geq 0$. Some
authors, for example \cite[Sec.1.1, p.550]{CST} distinguish the bid and ask
side of the order book by attaching a negative sign to the bid volume.} we
have $V(t,x)=\overline{V}(t,x)+\underline{V}(t,x)$.

\begin{definition}
The \emph{best ask process} $\alpha$ is given by 
\begin{equation}  \label{alpha}
\alpha(t) := \inf\{ x\in\mathbb{R}: \overline V(t,x) > 0 \},\quad t> 0
\end{equation}
and the \emph{best bid process} $\beta$ is given by 
\begin{equation}  \label{beta}
\beta(t) := \sup\{ x\in\mathbb{R}: \underline V(t,x) > 0 \},\quad t> 0
\end{equation}
\end{definition}
\begin{remark}
It is shown in Hubalek, Krühner and Rheinländer \cite{HKR-SPDE} that
the relative volume random field $v\left( t,x\right) :=V\left(
t,x+W_{t}\right) $ is a weak solution (in an appropriate sense) of the SPDE%
\begin{align}
dv(t,x)& =\left( \frac{1}{2}\partial _{x}^{2}v(t,x)+g(x)\right) dt+\partial
_{x}v(t,x)dW_{t},  \label{dv(t,x)} \\
v(0,x)& =0,\quad v(t,0)=0,\quad\forall(t,x)\in\mathbb R_+\times\mathbb R,
\end{align}%
and $V$ can be expressed in terms of Brownian local time $L^{y}$ at the
level $y$ as 
\begin{equation}
V(t,x)=\int_{\mathbb{R}}\left( L_{t}^{x-y}-L_{\sigma (t,x)}^{x-y}\right)
g(y)dy,  \label{V-Wiener specific}
\end{equation}
which follows from (\ref{volume function}) by the occupation times formula, cf.\ \cite[Corollary VI.1.6]{RY}.
\end{remark}
\subsection{General order placement}

In view of (\ref{V-Wiener specific}), we propose for a finite Borel measure $%
G$ with $G(\{0\})=0$ the order book process 
\begin{equation}
V(t,x)=\int_{\mathbb{R}}\left( L_{t}^{x-y}-L_{\sigma (t,x)}^{x-y}\right)
G(dy).  \label{V-Wiener}
\end{equation}

The decomposition into bid and ask (\ref{bid-ask}) and the definitions for
best bid and ask (\ref{alpha}) and (\ref{beta}) apply unchanged also to the
model with general order placement.

Of particular importance for the analysis is the case when order placements
are not absolutely continuous, but occur only at a fixed distance $\mu >0$
from the center. This corresponds to choosing 
\begin{equation}
G(dx)=\delta _{-\mu }(dx)+\delta _{\mu }(dx),
\end{equation}%
with $\delta _{\pm \mu }$ denoting the Dirac distribution at $\pm \mu $, and
leads to 
\begin{equation}
\underline{V}(t,x)=L_{t}^{x+\mu }-L_{\sigma (t,x)}^{x+\mu },\quad \overline{V%
}(t,x)=L_{t}^{x-\mu }-L_{\sigma (t,x)}^{x-\mu },  \label{funda}
\end{equation}%
where $L_{t}^{x}$ denotes the Brownian local time at level $x$.

These definitions can be motivated by the analogous notions for the discrete
order book model from~\cite{HKR-DISCRETE}, which can be studied by
elementary counting of single orders of size one.

While this basic model is a gross simplification, it nevertheless gives some
insight about the classification of trading times, and leads to new
mathematical results regarding order avalanches which have been studied
before in the context of Parisian options, see \cite{DW}. 
Moreover, (\ref{funda}) can be considered as \emph{fundamental solution} 
or \emph{Green's function} for the SPDE~(\ref{dv(t,x)}) with general order placement
intensity~$g$.

\section{Trading times -- definition and classification}

\label{s:tradingtimes} 
In this section we start with a pathwise analysis of trading times on the
Wiener space. For this we consider a fixed $w\in \mathbf{W}$ that admits a continuous local
time function $L$, i.e.\ $L:[0,\infty )\times \mathbb{R}\rightarrow \mathbb{R}_{+}$
such that $\int_{0}^{t}1_{\{w(s)\in A\}}ds=\int_{A}L_{t}^{x}dx$ for any
Borel set $A\in \mathcal{B}(\mathbb{R})$, $t\geq 0$. 
\begin{remark}
To emphasize the pathwise nature of the results in this section 
we should write $L_t^x(w)$ instead of $L_t^x$
and similarily $\underline{V}(t,x,w)$, $\overline{V}(t,x,w)$ etc., but for better
readability we omit the $w$ in the notation.
\end{remark}
We define in complete analogy to
Equation~(\ref{funda}), but now for the single path $w$, 
\begin{equation}
\underline{V}(t,x):=L_{t}^{x+\mu }-L_{\sigma (t,x)}^{x+\mu },\quad \overline{%
V}(t,x):=L_{t}^{x-\mu }-L_{\sigma (t,x)}^{x-\mu },
\end{equation}%
$\alpha (t):=\inf \{x>w(t):\overline{V}(t,x)>0\}$ and $\beta (t):=\sup
\{x<w(t):\overline{V}(t,x)>0\}$ for $t>0$, $x\in \mathbb{R}$.

The trading times are exactly the times when the path $w$ hits the best ask $%
\alpha$, resp.\ the best bid $\beta$.

\begin{definition}
\label{trading time points} We define the set of \emph{ask} \emph{trading
times} $\overline{\Theta }$ and the set of \emph{bid trading times} $%
\underline{\Theta }$ by 
\begin{equation}
\overline{\Theta }:=\{t\geq 0:w(t)=\alpha (t)\},\quad \underline{\Theta }%
:=\{t\geq 0:w(t)=\beta (t)\}.
\end{equation}
\end{definition}

\begin{remark}
In the following we shall focus on the ask side and simply write $\Theta$
for $\overline\Theta$. The corresponding definitions and results for the bid
side are completely analogous.
\end{remark}

To classify trading times let us introduce the last and next trading time.

\begin{definition}\label{d:first and last trading time}
The \emph{last trading time} before~$t$ is 
\begin{equation}
\Upsilon(t):=\sup\{s\in[0,t):\mbox{$s\in\Theta$ or $s=0$}\},
\end{equation}
the \emph{next trading time} after~$t$ is 
\begin{equation}
\Xi(t) := \inf\{s\in(t,\infty):\mbox{$s\in\Theta$ or $s=\infty$}\}.
\end{equation}
\end{definition}

We start classifying different trades into those trades where the best ask
increases (Type~{I}) and those where the best ask decreases (Type~{II}). By
convention we consider the first trade to be of Type~{II}.

\begin{definition}
The set of \emph{Type~{I} trades} is defined by 
\begin{equation}
\Theta _{\mathrm{I}}:=\{t\in \Theta :\alpha (\Upsilon (t))\leq \alpha
(t),\Upsilon (t)\neq 0\}
\end{equation}%
the set of \emph{Type~{II} trades} is 
\begin{equation*}
\Theta _{\mathrm{II}}:=\{t\in \Theta :%
\mbox{$\alpha(\lasttrade(t)) >
\alpha(t)$ or $\lasttrade(t)=0$}\}.
\end{equation*}
\end{definition}

Schematic illustrations for a Type~{I} resp.\ Type~{II} trade are given in
Fig.~\ref{Fig:AskIs} resp.\ in Fig.~\ref{Fig:AskIIs}.

For a finer classification, we distinguish the cases where trades accumulate
(a) before and after $t$, (b) before but not after $t$, (c) after but not
before $t$, and (d) isolated trades. Thus a priory we have eight types.

\begin{definition}
\begin{equation}
\begin{array}{ll}
\Theta_\mathrm{I_a} & :=\{t\in \Theta_\mathrm{I}: \Upsilon(t) = t = \Xi(t)
\}, \\ 
\Theta_\mathrm{I_b} & :=\{t\in \Theta_\mathrm{I}: \Upsilon(t) = t < \Xi(t)
\}, \\ 
\Theta_\mathrm{I_c} & :=\{t\in \Theta_\mathrm{I}: \Upsilon(t) < t = \Xi(t)
\}, \\ 
\Theta_\mathrm{I_d} & :=\{t\in \Theta_\mathrm{I}: \Upsilon(t) < t < \Xi(t)
\},%
\end{array}
\quad 
\begin{array}{ll}
\Theta_\mathrm{II_a} & :=\{t\in \Theta_\mathrm{II}: \Upsilon(t) = t = \Xi(t)
\}, \\ 
\Theta_\mathrm{II_b} & :=\{t\in \Theta_\mathrm{II}: \Upsilon(t) = t < \Xi(t)
\}, \\ 
\Theta_\mathrm{II_c} & :=\{t\in \Theta_\mathrm{II}: \Upsilon(t) < t = \Xi(t)
\}\quad\text{and} \\ 
\Theta_\mathrm{II_d} & :=\{t\in \Theta_\mathrm{II}: \Upsilon(t) < t < \Xi(t)
\}.%
\end{array}%
\end{equation}
\end{definition}

However, we shall see in Proposition~\ref{no-IIab} that Type~{IIa} and Type~{%
IIb} trades do not exist, and, then again in a stochastic setup, in
Proposition~\ref{no isolated trades} that the probability for isolated
trades (i.e.\ Type~{Id} and Type~{IId}) is zero. The only trades that do
occur with positive probability are Type~{Ia}, Type~{Ib}, Type~{Ic}, and
Type~{IIc}.

A schematic illustration for various types of trades is given 
in Fig.~\ref{Fig:Paul} on page~\pageref{Fig:Paul}.
\begin{proposition}\label{no-IIab} 
Type~{II} trades do not accumulate from the left, i.e., we
have $\Theta _{\mathrm{II_{a}}}=\emptyset $ and 
$\Theta _{\mathrm{II_{b}}}=\emptyset $ for all $w\in \mathbf{W}$.
\end{proposition}

\begin{proof}
Let $t\in \Theta_\mathrm{II}$ and assume that $\Upsilon (t)=t$. Then $\alpha
\left( \Upsilon (t)\right) =\alpha \left( t\right) $ and hence $t\notin
\Theta_\mathrm{II}$. Therefore, $t\in \Theta _{\mathrm{II_{c}}}\cup \Theta _{%
\mathrm{II_{d}}}$.
\end{proof}

\begin{remark}
Note that if $w_{0}=0$, and the order book is initially empty, then the
first trade will be a Type~{II} trade by definition. Moreover, all Type~{II}
trades afterwards happen at levels below or equal the last trading level,
whereas Type~{I} trades happen at levels higher than the last trade. After
the first trade, we have the following succession of trades: at first, there
are in every $\varepsilon $-interval infinitely many Type~{Ia} trades, until
a downward excursion from $\alpha $, i.e.\ an upward excursion of $\alpha -W$
from zero.
\end{remark}

\begin{figure}[H]
\includegraphics[width=0.9\textwidth]{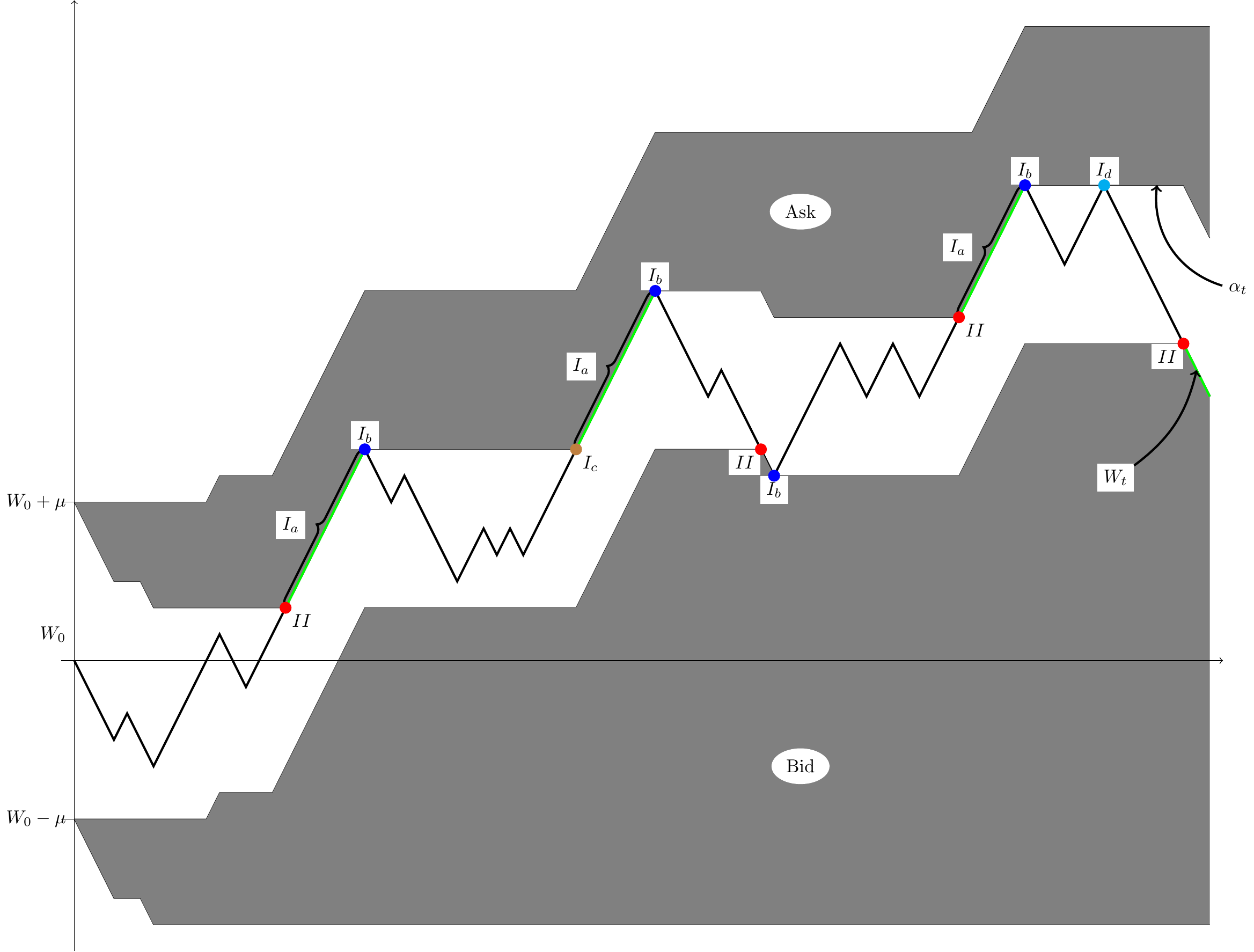}%
\caption{A schematic illustration of different types of trading times}
\label{Fig:Paul}
\end{figure}

We introduce an alternative representation for the best ask process $\alpha$. 
This will be used in the next sections for characterising the time to next
trade in a probabilistic way. Denote by 
\begin{equation}
w^{\ast}(s,t):=\sup\left\{ w(r):r\in\left[ s,t\right] \right\}, \quad
w_{\ast}(s,t):=\inf\left\{ w(r):r\in\left[ s,t\right] \right\}
\end{equation}
the running maximum respectively minimum of the path $w$ in the interval $%
[s,t]$.

\begin{definition}
\label{type II trades} Let 
\begin{equation}
\Gamma_s:=\inf\{t\geq s:w(t)=w_{\ast}(s,t)+\mu\},\quad \Psi_{s}
:=\inf\{t\geq s:w(t)=w^{\ast}(s,t)-\mu\}
\end{equation}
and define for any $n\geq0$ the times $\tau_{0}:=0$ and 
\begin{equation}
\tau_{n+1}:=%
\begin{cases}
\infty & \tau_{n}=\infty \\ 
\Gamma_{\tau_{n}} & n\text{ even, }t_{n}<\infty \\ 
\Psi_{\tau_{n}} & \text{otherwise.}%
\end{cases}%
\end{equation}
\end{definition}

We start with a small observation for the time-points $\tau_{n}$, $n\in%
\mathbb{N}$.

\begin{lemma}
The sequence $(\tau_{n})_{n\in\mathbb{N}}$ is strictly increasing until
reaching $\infty$ and it has no finite accumulation point.
\end{lemma}

\begin{proof}
Let $n\in\mathbb{N}$ such that $\tau_{n}\neq\infty$. By continuity of $w$ there
is $\eta>0$ such that $|w(t)-w(\tau_{n})|<\mu/2$ for any $t\in\lbrack \tau_{n}%
,\tau_{n}+\eta]$. Then, $w^{\ast}(\tau_{n},\tau_{n}+\eta)-w_{\ast}(\tau_{n},\tau_{n}%
+\eta)<\mu$ and hence $\Psi_{\tau_{n}},\Gamma_{\tau_{n}}>\tau_{n}+\eta$. Consequently,
$\tau_{n+1}>\tau_{n}+\eta$.

Now assume by contradiction that $\tau_{n}\nearrow t$ for some $t\in(0,\infty)$.
By continuity of $w$, there is $\eta>0$ such that $|w(t)-w(s)|<\mu/2$ for any
$s\in\lbrack t-\eta,t]$. Moreover, there is $n\in\mathbb{N}$ such that
$|\tau_{n}-t|\leq\eta$ and hence we have $t>\tau_{n+1}>\tau_{n}+\eta\geq t$. A contradiction.
\end{proof}
We can now identify the behavior of the best ask process $\alpha$ in terms
of the times $(\tau_n)_{n\geq 0}$.

\begin{proposition}
\label{structure alpha} We have 
\begin{equation*}
\alpha(t) = 
\begin{cases}
w_{\ast}(\tau_{n},t)+\mu; & \tau_{n}\leq t<\tau_{n+1},~n\text{ \textrm{even}}
\\ 
w^{\ast}(\tau_{n},t); & \tau_{n}\leq t<\tau_{n+1},~n\text{ \textrm{odd.}}%
\end{cases}%
\end{equation*}
for any $t> 0$. Moreover, $\alpha$ is a continuous function of finite
variation which is non-increasing on $\{t: \exists n\in\mathbb{N}: \tau_n
\leq t < \tau_{n+1},n\text{ even}\}$ and non-decreasing on the compliment.
\end{proposition}
\begin{proof}
  Define \[
\gamma(t):=%
\begin{cases}
w_{\ast}(\tau_{n},t)+\mu; & \tau_{n}\leq t<\tau_{n+1},~n\text{ \textrm{even}}\\
w^{\ast}(\tau_{n},t); & \tau_{n}\leq t<\tau_{n+1},~n\text{ \textrm{odd.}}%
\end{cases}
\]
for any $t>0$. We first show that $\gamma$ is continuous. Clearly, $\gamma$ is c\`adl\`ag and it is continuous outside $\{\tau_{n}:n\geq0\}$ by
definition. Let $n\geq1$ with $\tau_{n}\neq\infty$.

\underline{Case 1}: $n$ is even. By definition we have $\tau_{n}=\Psi_{\tau_{n-1}}$ and hence
$w(\tau_{n})+\mu=w^{\ast}(\tau_{n-1},\tau_{n})$.
\begin{align*}
\gamma(\tau_{n}-)  &  =\lim_{t\nearrow \tau_{n}}\gamma(t)=\lim_{t\nearrow \tau_{n}%
}w^{\ast}(\tau_{n-1},t)=w^{\ast}(\tau_{n-1},\tau_{n})\\
&  =w(\tau_{n})+\mu=w_{\ast}(\tau_{n},\tau_{n})+\mu=\gamma(\tau_{n}).
\end{align*}

\underline{Case 2}: $n$ is odd. This is proved analogously like the even case.

Thus $\gamma$ is a continuous function. Next we show that $\alpha=\gamma$.
Now, let $t\in(\tau_0,\tau_1]$. Then, $\gamma(t) = w_*(t_0,t)+\mu$ and $\{w(s):s\in [0,t]\} = [w_*(0,t),w^*(0,t)]$. Hence 
$L_t^u>0$ for Lebesgue almost any $u\in[w_*(0,t),w^*(0,t)]$ and $L_t^u=0$ for any $u\in\mathbb R\backslash [w_*(0,t),w^*(0,t)]$. Consequently,
  $$ \alpha(t) = \inf\{x>w(t) : V(t,x)>0 \} = w_*(0,t) + \mu = \gamma(t). $$
  
Now let 
\begin{align*} 
   I:=\{n\in\mathbb N: \forall t\in (\tau_n,\tau_{n+1}]: \alpha(t)=\gamma(t)\}. 
\end{align*}
Let $n\in I\cup\{0\}$ and $t\in [\tau_{n+1},\tau_{n+2}]$.

\underline{Case 1}: $n+1$ is even. Then, $\gamma(t) = w^*(\tau_{n+1},t)$ and 
 $w(t) \leq w_*(\tau_{n+1},\tau_{n+2})+\mu$ by definition of $\gamma$ and $\tau_{n+2}$. Consequently,
  $$ w^*(\tau_{n+1},\tau_{n+2}) = w_*(\tau_1,\tau_2)+\mu $$
 and hence
  $$ V(t,x) = \left(V(\tau_{n+1},x)+L_t^{x-\mu}-L_{\tau_{n+1}}^{x-\mu}\right)\mathbf{1}_{\{w^*(\tau_{n+1},t)<x\}}\geq V(\tau_{n+1},x)\mathbf{1}_{\{w^*(\tau_{n+1},t)<x\}}.  $$
 Lemma \ref{l:depths of lob} yields $\alpha(t) = w^*(\tau_{n+1},t) = \gamma(t)$.
 
\underline{Case 2}: $n+1$ is odd. This works similar and we get $n+1\in I$.

By induction $\mathbb N= I\cup\{0\}$ which yields the claim.
\end{proof}
\begin{corollary}
\label{type II trades as stopping times} $\tau_1$ is the first (ask) trade
and $\tau_{2n-1}$ denotes the $n$-th Type~{II} trade (in the ask order book)
for any $n\in\mathbb{N}$.
\end{corollary}
In view of this corollary we define, for mathematical convenience, 
the first trade to be of Type~II.
\begin{proof}
  The first statement is immediate from the definition and the second 
  statement follow immediately from Proposition \ref{structure alpha}.
\end{proof}
\begin{remark}
In this section we have not used any specific properties of Brownian
motion. In fact, we solely argued from the existence of a continuous
occupation density which exists a.s.\ for many processes. 

Let $X$ be any continuous semimartingale such that its quadratic variation
is given by $[X,X](t) = \int_0^t c(X(s)) ds$ where $c:\mathbb{R}%
\rightarrow(0,\infty)$ is a continuous function. \cite[VI.1.7]{RY} yields
that it has local time $L$ which is continuous in time and c\`adl\`ag in its
space variable. \cite[Corollary VI.1.6]{RY} yields that for any Borel set $%
A\in\mathcal{B}(\mathbb{R})$ we have 
\begin{equation*}
\int_0^t 1_{\{X(s)\in A\}} ds = \int_0^t \frac{1_{\{X(s)\in A\}}}{c(X(s))}
d[X,X](s) = \int_{A} \frac{L_t^x}{c(x)} dx
\end{equation*}
and, hence, $X$ has occupation density $\rho^t_x = \frac{L_t^x}{c(x)}$, $x\in%
\mathbb{R}$, $t\geq0$. In particular, if its local time posses a continuous
version, then so does its occupation density.

For more details on occupation densities see \cite{geman.horowitz.80}.
\end{remark}
\begin{remark}
Throughout this section we worked with the specific Dirac order placement.
However, a close inspection of the arguments reveals that this is not
strictly necessary to obtain the preceeding results. If orders are placed
with respect to some measure $G$ instead, as in Equation \eqref{V-Wiener},
and $0\notin \mathrm{supp}(G)$, then defining $\mu := \inf\mathrm{supp}(G|_{\mathcal{B}(\mathbb{R}_+)})$ 
allows to obtain the same results as presented
in this section as long as $\mu>0$.
\end{remark}
\section{Analysis of the Brownian order book with Dirac order placement}
\subsection{Characterizing trading times via a doubly reflected Brownian
motion}
We return to our stochastic setup as in Section~\ref{s:ac order placement}.
Our aim is now to characterize trading times via a doubly reflected Brownian
motion. The results from the preceding section hold almost surely
by the occupation times formula \cite[Theorem VI.1.6]{RY} and \cite[Theorem VI.1.7]{RY}.

\begin{remark}
\label{r: tau stopping times} Firstly, we observe that $(\tau_n)_{n\geq 0}$ from Definition~\ref{type II trades}
is an increasing sequence of stopping times.
\end{remark}

Up to here we have essentially gathered pathwise properties which do not
rely on the specific structure of the Brownian motion except for the
continuous sample path property and the existence of a continuous occupation
density. This, however, holds for many other processes as well, cf.\ \cite[Theorem IV.76, Corollary IV.2]{protter.04}. For the rest of this section we
consider features of trading times which appear to be more specific to the
Brownian motion.

\begin{definition}
\label{d:DRBM} Let $\mu>0$. A $[0,\mu]$-valued stochastic process $X$ is
called a \emph{doubly reflected Brownian motion} if 
\begin{equation*}
f(X(t)) - \int_0^t \frac{1}{2}f^{\prime \prime }(X(s)) ds,\quad t\geq 0
\end{equation*}
is a martingale for any twice continuously differentiable function $f:[0,\mu]\rightarrow\mathbb{R}$ with $f^{\prime }(0)=0=f^{\prime }(\mu)$.
\end{definition}

Recall that \cite[Theorems 8.1.1, 4.5.4]{ethier.kurtz.86} yield that such a
process exist and \cite[Theorem 4.4.1]{ethier.kurtz.86} yields that its
process law is uniquely determined by its initial distribution $P^{X(0)}$.

\begin{theorem}
\label{doubly reflected BM}The process $\alpha-W$ is a doubly reflected
Brownian motion on the interval $[0,\mu]$ with $(\alpha-W)(0)=\mu$.
Moreover, we have $\Theta=\{t:(\alpha-W)(t)=0\}$.

Clearly, $W-\beta $ is another doubly reflected Brownian motion on the
interval $[0,\mu ]$ and $\underline{\Theta }=\{t:(W-\beta )(t)=0\}$. 
\end{theorem}

\begin{proof}
Let $f:[0,\mu]\rightarrow\mathbb R$ be a twice continuously differentiable 
function with $f'(0)=0=f'(\mu)$ and define $R:=W-\alpha$. Let
 \begin{align*} 
 I:=\Bigg\{n\in\mathbb N: & \textstyle E\left[f(R(\tau_n+1))-\frac12\int_{\tau_n}^{\tau_{n+1}}f''(R(s))ds \middle\vert\mathcal F_{\eta}\right] = 
 f(R(\eta))-\int_{\tau_n}^{\eta}\frac{1}{2}f''(R(s))ds \\ 
  &\text{ for any stopping time }\eta\in[\tau_n,\tau_{n+1}]\Bigg\}
 \end{align*}
Let $n\in \mathbb N$ be even and define $X(t) := W(t+\tau_n)-W(\tau_n)$ 
and $X_\ast(t):=\inf\{t\geq0:X(t)\}$. Then $X$ is a Brownian motion which 
is independent of $\mathcal F_{\tau_n}$ and Proposition~\ref{structure alpha} 
yields that $R(\eta+\tau_n) = R(\tau_n)-X(\eta)+X_\ast(\eta)$ for any random 
time $\eta$ which is bounded by $\Delta\tau_n:=\tau_{n+1}-\tau_n$. Moreover, the 
law of $X_\ast-X$ coincides with the law of $-|B|$ for some Brownian motion $B$, which follows from
a well-known result of L\'evy, see for example \cite[Thm.VI.2.3, p.240]{RY},  and hence 
$$ f(R(\tau_n)+(X_\ast-X)^{\Delta\tau_n}(t)) + \int_0^{t\wedge\Delta\tau_n} \frac{1}{2} f''(R(\tau_n)+(X_\ast-X)(s)) ds,\quad t\geq 0$$
is a martingale. Hence, $n\in I$.
For odd $n\in\mathbb N$ similar arguments show that $n\in I$ and thus $I=\mathbb N$. 
The tower property yields that $(R(t) - \int_0^t \frac{1}{2}f''(R(s)) ds)_{t\geq 0}$ is a 
martingale and hence $R$ is an $[0,\mu]$-valued process with $[R](t)=[W](t)=t$ and reflecting 
boundaries and hence a doubly reflected Brownian motion, cf.~\cite[p. 366]{ethier.kurtz.86}.
\end{proof}Next we show that there are no isolated trades, i.e.\ there are
no trades of Type~{Id} or Type~{IId}.

\begin{corollary}
\label{no isolated trades}We have no isolated trades, i.e., 
$P(\Theta_{\mathrm{I}_d}\cup\Theta_{\mathrm{II}_d}=\emptyset)=1$. 
In particular, we have 
\begin{align*}
P(\Theta_{\mathrm{II}} = \Theta_{\mathrm{II}_c}) &= 1.
\end{align*}
\end{corollary}

\begin{proof}
By Theorem~\ref{doubly reflected BM} we have to show that a doubly
reflected Brownian motion on $[0,\mu]$ has $P$-a.s.\ no isolated zeros. Using
the construction in~\cite{KS}, Section 2.8.C, we see that this is equivalent to show that a standard Brownian motion
has $P$-a.s.\ no isolated times in the set $\{2z\mu:z\in\mathbb{Z}\}$. This is
a consequence of~\cite{KS}, Theorem 9.6, Chapter 2.
\end{proof}

\subsection{Stopping times and trading times}

So far we have defined trading times pathwise: $t\in \mathbb{R}^{+}$ is a
trading time for $w\in \mathbf{W}$ if $w(t)=\alpha (t,w)$. We say a random
time $\tau $ is a trading time, if $W_{\tau}=\alpha (\tau)$ a.s. Next, we will give
some examples of trading times which are also stopping times and we will
show that a stopping time which is a trading time is not of Type~{Ib}.

\begin{lemma}
\label{stopping times not 1b} Let $\tau$ be a stopping time such that $P(\tau\in\Theta)=1$. Then $\Xi(\tau)=\tau$ $P$-a.s. In particular, $P(\tau\in\Theta_{\mathrm{I}_b})=0$.
\end{lemma}

\begin{proof}
  Let $\epsilon >0$ and define $B(t):=W(t+\tau)-W(\tau)$, $t\geq0$. 
  Then, $B$ is a standard Brownian motion. Let $\sigma_\epsilon$ be the 
  time where $B$ attains its maximum on $[0,\epsilon]$. Then, 
  $\sigma_\epsilon$ is measurable and $B(\sigma_\epsilon)>0$ $P$-a.s. 
  Hence $\tau+\sigma_\epsilon$ is a trading time and, consequently, 
  $P(\nexttrade(\tau)=\tau)=1$.
\end{proof}Corollary \ref{type II trades as stopping times} together with
Remark \ref{r: tau stopping times} reveals that the Type~{II} trades can be
enumerated by stopping times. Lemma \ref{stopping times not 1b} shows that
the Type~{Ib} trades are not stopping times. This leaves the question
whether stopping times can be of Type~{Ia} or Type~{Ic}, which is, indeed,
the case.

First entry times of high levels are actually Type~{Ia} trades.

\begin{example}
Let $x > \mu$ and $\gamma_x:=\inf\{t\geq 0: w(t) = x\}$. Then we have 
\begin{equation*}
P(\gamma_x \text{ is a Type~{Ia} trade}) = 1.
\end{equation*}
\end{example}

\begin{proof}
  By Proposition~\ref{structure alpha} we have $\alpha(t) \leq w^*(0,t)$ and
  clearly $w(t)\leq \alpha(t)$ for any $t\in[\tau_1,\infty)$. Since 
  $$\tau_1 = \inf\{t>0: w(t) = w^*(0,t)-\mu\}$$
 we have $w^*(0,\tau_1)\leq \mu$. Hence, we have $\tau_1\leq \gamma_x$. Thus we get
  $$ w^*(0,\gamma_x)=x=w(\gamma_x) \leq \alpha(\gamma_x)\leq w^*(0,\alpha_x). $$
 and hence $\gamma_x\in\Theta$.
 
  Let $\epsilon\in(0,\gamma_x)$. Then, there is $s_\epsilon\in(\gamma_x-\epsilon,\gamma_x)$ 
  such that $w(s_\epsilon) = w^*(0,s_\epsilon)>\mu$. Hence, $\alpha(s_\epsilon)=w(s_\epsilon)$ and 
  we have $s_\epsilon\in\Theta$. This implies that $\lasttrade(\gamma_x) = \gamma_x$. 
  Lemma~\ref{stopping times not 1b} yields that $\nexttrade(\gamma_x)=\gamma_x$ $P$-a.s. Hence, $\gamma_x\in\Theta_{\mathrm{I}_a}$ $P$-a.s.
\end{proof}

\begin{example}\label{ExIc}
There is a stopping time $\eta$ such that 
\begin{equation*}
P(\eta \text{ is a Type~{Ic} trade}) = 1.
\end{equation*}
\end{example}

\begin{proof}
 For a stopping time $\eta$ define the new stopping times
  \begin{align*}
    \gamma_0(\eta) &:= \inf\{t\geq \eta: (\alpha-W)(t) = 0\}, \\
    \gamma_1(\eta) &:= \inf\{t\geq \eta: (\alpha-W)(t) = \mu/2\}, \\
    \gamma_2(\eta) &:= \inf\{t\geq \eta: (\alpha-W)(t) \in\{0,\mu\}\}.
  \end{align*}
  Clearly, $\gamma_j(\eta)$ is $P$-a.s.\ finite for any finite stopping time $\eta$, $j=0,1,2$. Moreover, $$P((\alpha-W)(\gamma_2(\gamma_1(\eta))) = 0) = 1/2$$ 
  by symmetry and the Markov property for any stopping time $\eta$. We have
$$ 
A(\eta) := \{ \gamma_2(\gamma_1(\gamma_0(\eta))):(\alpha-W)(\gamma_2(\gamma_1(\gamma_0(\eta)))) = 0 \} 
\subseteq \Theta_{\mathrm{I}_c}
$$
for any finite stopping time $\eta$. 
  
  Define $\eta_0 := \tau_2$ where $\tau_2$ is given in Definition \ref{type II trades}. 
  Observe that $w(\tau_2)\neq \alpha(\tau_2)$. Define recursively
  $$ \eta_{n+1} := 
   \begin{cases}
      \eta_n & \text{if $\alpha(\eta_n) = w(\eta_n)$}, \\
      \gamma_2(\gamma_1(\gamma_0(\eta_n))) & \text{otherwise.}
   \end{cases}
  $$
 Then, $(\eta_n)_{n\in\mathbb N}$ converges $P$-a.s.\ in finitely many steps. 
 Denote $\eta_\infty := \lim_{n\rightarrow\infty}\eta_n$. Clearly, $\alpha(\eta_\infty) = w(\eta_\infty)$ $P$-a.s. 
 Moreover, denote $\eta_- := \eta_{\{\sup\{n\in\mathbb N:\eta_n\neq \eta_\infty\}\}}$. 
 Then, we have $\gamma_2(\gamma_1(\gamma_0(\eta_-))) = \eta_\infty$. Consequently, 
  $$ P(\eta_\infty \in \Theta_{\mathrm{II}_c}) = 1. $$
\end{proof}

\section{Trading excursions}

\subsection{The trading excursion process}

Theorem~\ref{doubly reflected BM} shows that trading times correspond to the
zeroes of the Markov processes $X$ and $Y$ defined by 
\begin{equation}
X=\alpha -W,\quad Y=W-\beta .
\end{equation}
This allows to study trading times by using \emph{excursion theory}. Let us
recapitulate briefly the terminology and notation of excursion theory, for
background and more details we refer the reader to \cite[Ch.XII]{RY} and 
\cite{Blu}.

For $w\in\mathbf{W}$ define 
\begin{equation}
R(w)=\inf\{t>0:w(t)=0\}.
\end{equation}
Let $U^+$ denote all nonnegative functions $w$ such that $0<R(w)<\infty$,
let $\delta$ denote the function that is identically zero, and set $U^+_\delta=U^+\cup\{\delta\}$, and let $\mathcal{U}^+_\delta$ denote the trace
of the Borel $\sigma$-field on $\mathbf{W}$ in $U^+_\delta$.

First we note that $X$ is a continuous semi-martingale, namely doubly
reflected Brownian motion on $[0,\mu ]$. Thus it admits a local time at zero
that satifies the Tanaka formula, 
\begin{equation}
L_{t}(X)=|X_{t}|-|X_{0}|-\int_{0}^{t}\sgn(X_{s})dX_{s},\quad t\geq 0.
\end{equation}
Consider the inverse local time process, 
\begin{equation}
\tau _{s}(X)=\inf \{t\geq 0:L_{t}(X)\geq s\},\quad s>0.
\end{equation}

\begin{definition}
The \emph{trading excursion process for the ask-side} is the process $(\overline{e}_{s},s>0)$, 
i.e. the zero-excursion process for $X$. The \emph{trading excursion process for the bid-side} is the 
process $(\underline{e}_{s},s>0)$ i.e. the zero-excursion process for $Y$.
\end{definition}

This means, that $\overline{e}$ and $\underline{e}$ are defined 
on $\Omega\times\mathbb R_+$ and take values in $U_\delta^+$ as
follows, see \cite[Def.XII.2.1, p.480]{RY}:

\begin{enumerate}
\item If $\Delta \tau _{s}(X)>0$, then $\overline{e}_{s}(w)$ is the map 
\begin{equation}
r\mapsto \overline{e}_{s}(r,w)=I_{[r\leq \Delta \tau _{s}(w)]}X_{\tau
_{s-}(w)+r}(w),
\end{equation}

\item if $\Delta \tau _{s}(X)=0$, then $\overline{e}_{s}(w)=\delta .$
\end{enumerate}

Thus $\overline{e}$ and $\underline{e}$ take values in the function space $U_\delta^+$. Illustrations for the trading excursion process are given in
Figures~\ref{Fig:AskIs},\ref{Fig:AskIx},\ref{Fig:AskIe} for a Type~{I} trade
resp.\ in Figures~\ref{Fig:AskIIs},\ref{Fig:AskIIx},\ref{Fig:AskIIe} for a
Type~{II} trade.

\begin{theorem}
The bid and ask trading excursion processes are Poisson point processes.
\end{theorem}

\begin{proof}
The previous proposition says that ask trading excursions for $W$ correspond to excursions from zero for $\alpha-W$. The process $\alpha-W$ is a doubly reflected Brownian motion, which is a Markov process.
We can apply \cite[Thm.3.18, p.95]{Blu}. The same argument holds for $W-\beta$.
\end{proof}
\begin{figure}[H]
\includegraphics{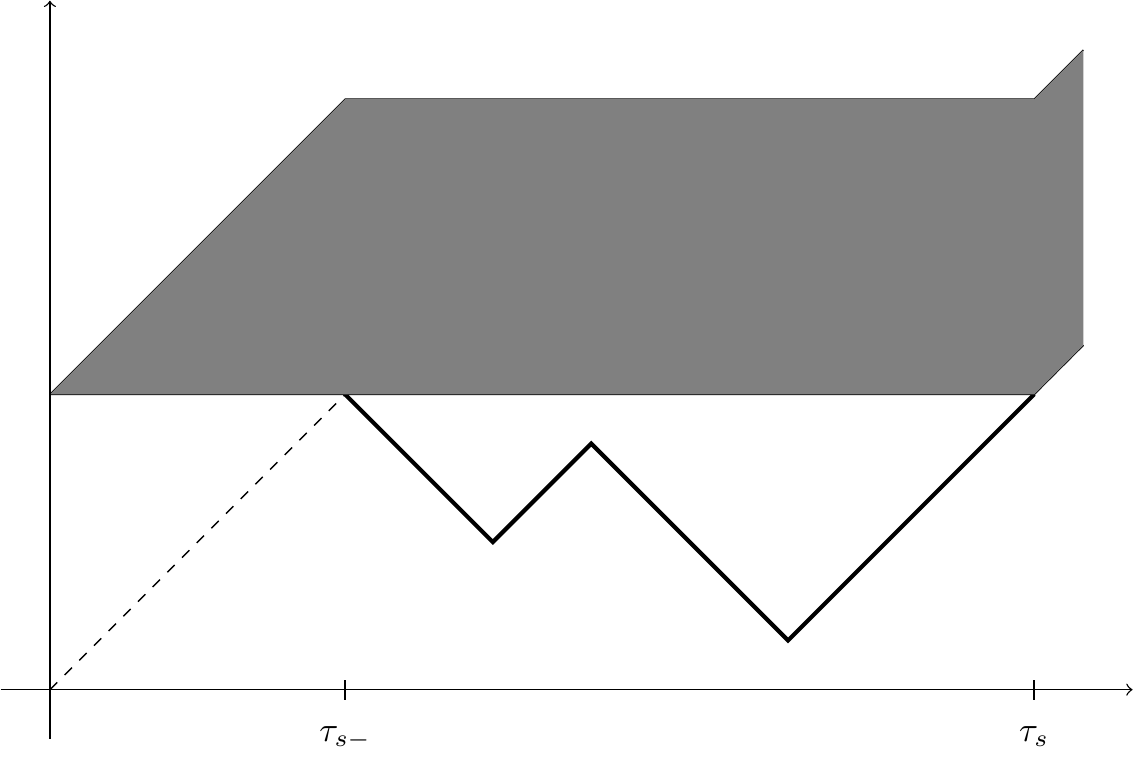}
\caption{Ask trade of Type I}
\label{Fig:AskIs}
\end{figure}
\begin{figure}[H]
\includegraphics{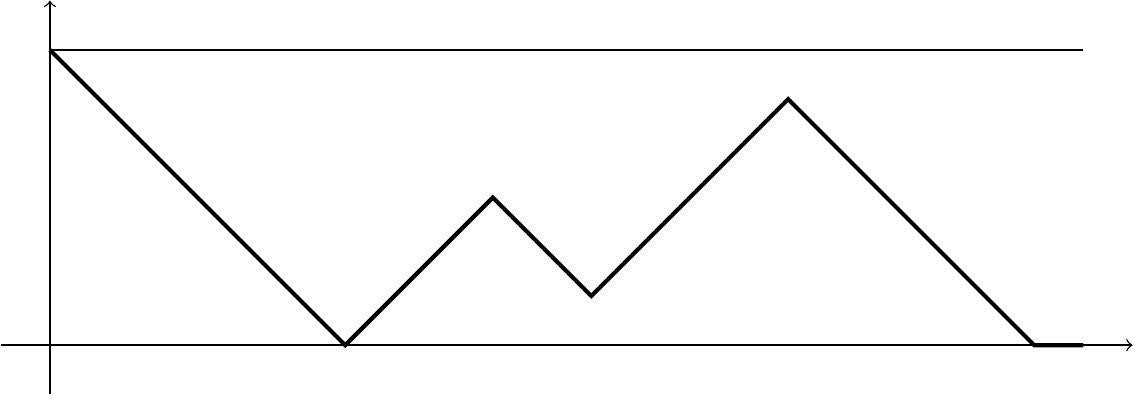}
\caption{Corresponding path of $\protect\alpha-W$}
\label{Fig:AskIx}
\end{figure}
\begin{figure}[H]
\includegraphics{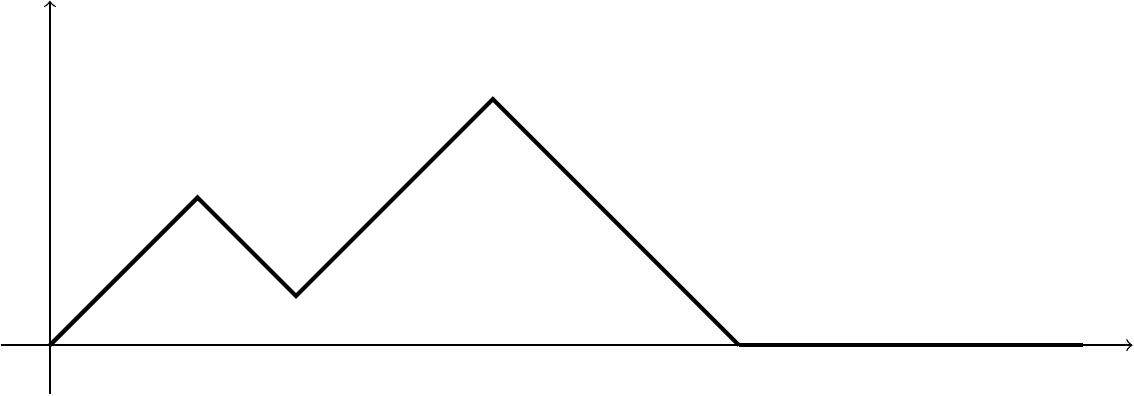}
\caption{Corresponding excursion $\overline{e}_s$}
\label{Fig:AskIe}
\end{figure}

\begin{figure}[H]
\includegraphics{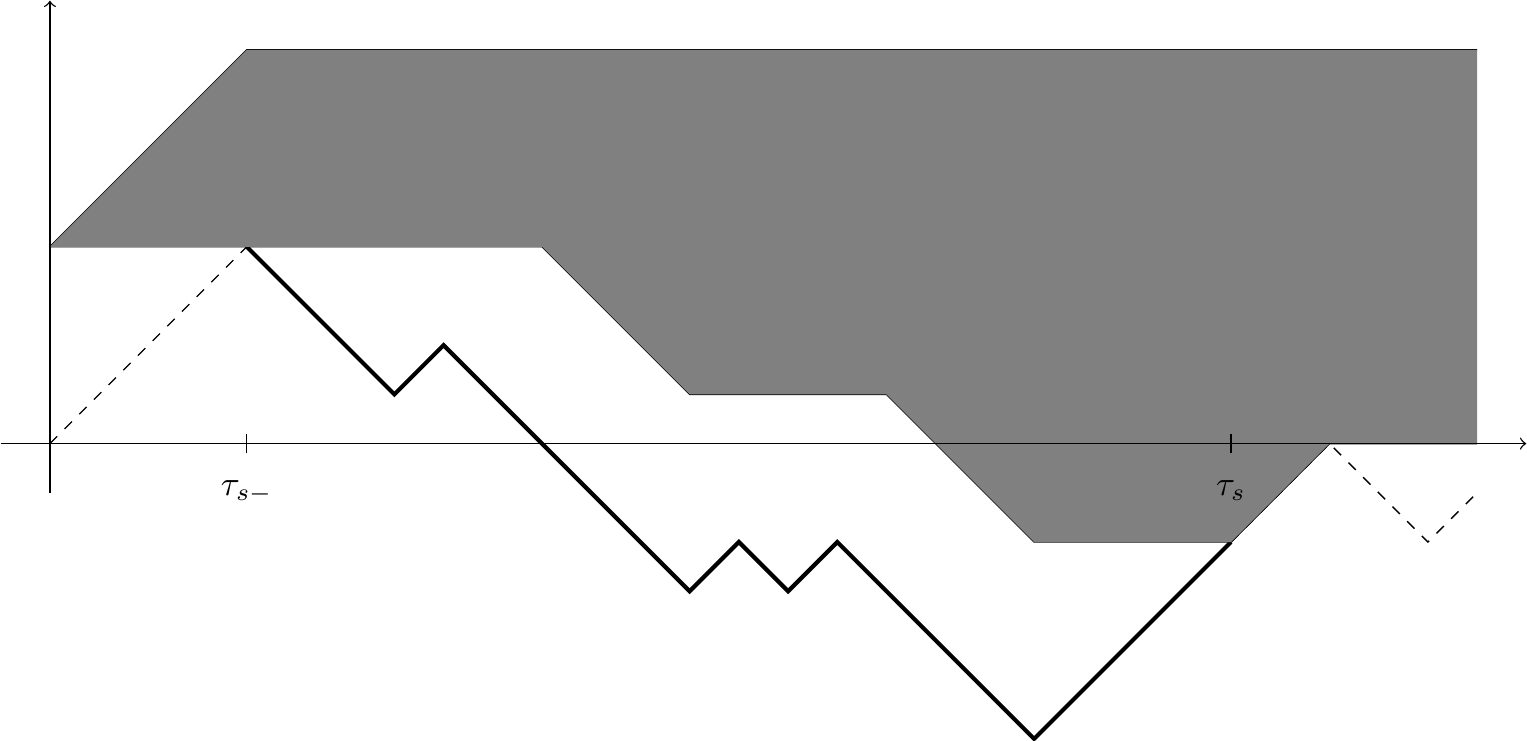}
\caption{Ask trade of Type II}
\label{Fig:AskIIs}
\end{figure}
\begin{figure}[H]
\includegraphics{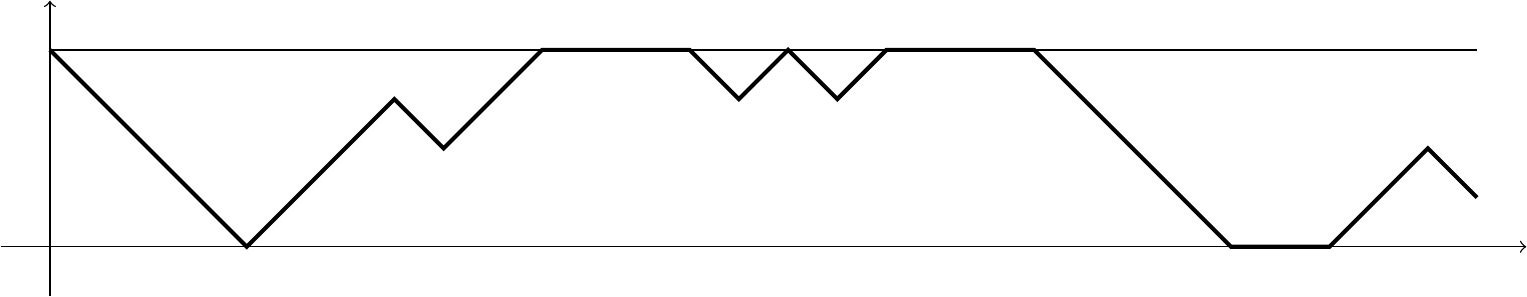}
\caption{Corresponding path of $\protect\alpha-W$}
\label{Fig:AskIIx}
\end{figure}
\begin{figure}[H]
\includegraphics{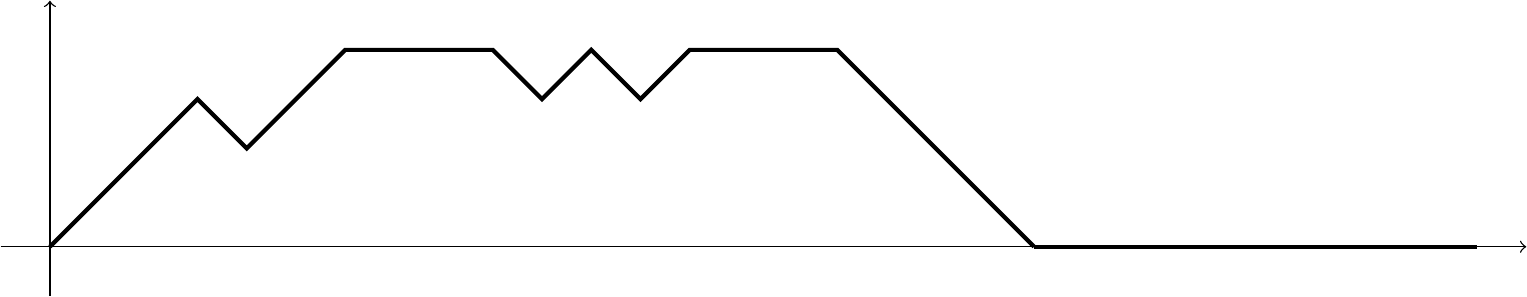}
\caption{Corresponding excursion $\overline{e}_s$}
\label{Fig:AskIIe}
\end{figure}

\subsection{Description of the trading excursion measure\label{x-measure}}
For any Poisson point process there exists an intensity measure.
\begin{definition}
Let us denote the intensity measures for the bid and ask excursion processes
by $\underline{n}$ and $\overline{n}$.
\end{definition}
The measures $\nbid$ and $\nask$ are $\sigma$-finite measures on $\mathcal{U}_\delta^+$,
and satisfy
\begin{equation}
\nbid(\Gamma)=\frac{1}{t}E[\underline{N}_t^\Gamma],\quad
\nask(\Gamma)=\frac{1}{t}E[\overline{N}_t^\Gamma],\quad t>0,
\end{equation}
where
\begin{equation}
\underline{N}_t^\Gamma=\sum_{0<s<t}I_{\Gamma}(\underline{e}_s),\quad
\overline{N}_t^\Gamma=\sum_{0<s<t}I_{\Gamma}(\overline{e}_s),\quad \Gamma\in \mathcal U_\delta^+.
\end{equation}

\begin{remark}
In the following we focus on $\underline{n}$ because $\overline n = \underline n$.
\end{remark}

Let us recall a convenient notation from Williams \cite[Sec.5.0, p.49]{PWM},
for the integral of a measurable function $F:U_\delta^+:\to\mathbb R$ with respect to the
measure $\underline{n}$ and a set $\Gamma\in\mathcal{U}_\delta^+$, 
\begin{equation}
\underline{n}(F)=\int F(u)\underline{n}(du),\quad \underline{n}(F;\Gamma)=\int_\Gamma F(u)\underline{n}(du).
\end{equation}
For better readability, we shall also write $\underline{n}[F]$ and $\underline{n}[F;\Gamma]$ instead of $\underline{n}(F)$ and $\underline{n}(F;\Gamma)$ when the expressions for $F$ or $\Gamma$ are more involved.
Furthermore, let us introduce for $x>0$ and $u\in C(\mathbb{R}_+;\mathbb{R})$
the hitting time 
\begin{equation}
T_x(u)=\inf\{t>0:u(t)=x\}.
\end{equation}
We can now give a description of the trading excursion measure that is
inspired by Williams' description of the Ito measure. Pick three independent
processes, namely two $\BES^3(0)$ processes $\rho$ and $\tilde\rho$, and a
standard Brownian motion $b$ (a $\BES^3(0)$ processes is a process whose law coincides with the law of $|B|$ where $B$ is a three dimensional standard Brownian motion starting in zero.). For all $x\in(0,\mu)$ we define a process $Z^x$
by 
\begin{equation}  \label{Zx}
Z^x=\left\{ 
\begin{array}{ll}
\rho_t & 0\leq t\leq T_x(\rho), \\ 
x-\tilde\rho_{t-T_x(\rho)} & T_x(\rho)<t\leq T_x(\rho)+T_x(\tilde\rho), \\ 
0 & t>T_x(\rho)+T_x(\tilde\rho),\end{array}
\right.
\end{equation}
and we define 
\begin{equation}  \label{Zmu}
Z^\mu=\left\{ 
\begin{array}{ll}
\rho_t & 0\leq t\leq T_\mu(\rho), \\ 
\mu-|b_{t-T_\mu(\rho)}| & T_\mu(\rho)<t\leq T_\mu(\rho)+T_\mu(|b|), \\ 
0 & t> T_\mu(\rho)+T_\mu(|b|).\end{array}
\right.
\end{equation}
Let us introduce length $R$ and height $H$ for excursions $u\in U_\delta^+$ by 
\begin{equation}
R(u)=\inf\{t>0:u(t)=0\},\quad H(u)=\sup\{u(t):0\leq t\leq R(u)\},
\end{equation}
and 
\begin{equation}
R_I(u)=R(u)I_{H(u)<\mu},\quad R_{II}(u)=R(u)I_{H(u)\geq\mu}.
\end{equation}
So, $R_I(u)$ is the length of a trading excursion that ends with a Type~{Ic}
trade, and zero otherwise, whereas $R_{II}(u)$ is the length of a trading
excursion that ends with a Type~{II} trade, and zero otherwise.

\begin{theorem}
\label{ThmWilT} For any $\Gamma\in\mathcal{U}_\delta^+$ 
\begin{equation}
\underline{n}(\Gamma)=\frac12\int_0^\mu
P[Z^x\in\Gamma]x^{-2}dx+\frac1{2\mu}P[Z^\mu\in\Gamma].
\end{equation}
\end{theorem}

\begin{proof}
This is a combination of two results about Williams' description of the Ito measure,
namely the excursion conditioned on a fixed height, see \cite[Thm.XII.4.5, p.499]{RY},
and the decomposition of the excursion straddling
the first hitting time of level $\mu$ as presented in  \cite[Prop.3.3, p.237]{Rog} and \cite[6.8~(a), p.75]{YY}.
\end{proof}

\begin{corollary}
\label{CorWilT} Let $F:C(\mathbb{R}_+;\mathbb{R})\to\mathbb{R}$ be a
non-negative measurable function. Then 
\begin{equation}
\underline{n}(F)=\frac12\int_0^\mu E[F(Z^x)]x^{-2}dx+\frac1{2\mu}E[F(Z^\mu)].
\end{equation}
The formula holds also true if $F$ is real- or complex-valued and $\underline{n}(|F|)<\infty$, or equivalently, if 
\begin{equation}
\int_0^\mu E[|F(Z^x)|]x^{-2}dx+E[|F(Z^\mu)|]<\infty.
\end{equation}
\end{corollary}

\begin{corollary}
We have 
\begin{equation}
\underline{n}[H\geq x]= \left\{ 
\begin{array}{ll}
\displaystyle\frac1{2x} & 0<x\leq\mu, \\ 
&  \\ 
\displaystyle0 & x>\mu,
\end{array}
\right.
\end{equation}
and 
\begin{equation}
\underline{n}[H\in
dx]=\frac1{2x^2}I_{(0,\mu)}(x)dx+\frac1{2\mu}\delta_\mu(dx).
\end{equation}
\end{corollary}

\section{The hyperbolic function table for intertrading times}

\subsection{The hyperbolic table under the trading excursion measure}

A trading excursion starts with a Ib trade. In this section we study the
time to the next trade after a trading excursion. Consider the time interval
from a Ib ask trade to the next trade. This is a trading excursion interval for $W$, and by Theorem~\ref{doubly reflected BM}, a zero excursion interval
for $\alpha-W$. So the time to the next trade is just the length of an excursion interval. We shall start with the trading excursion space $(U_\delta^+,\mathcal{U}_\delta^+,\underline{n})$ and then transfer the results to the probability space $(\Omega,\mathcal{F},P)$.

The time to the next trade for a trading excursion $u\in U_\delta^+$
is the length $R$ of the trading excursion, the type of the next trade
depends on the height $H$. If $H(u)<\mu$ the next trade is of Type~{I}, if $H(u)\geq\mu$ it is of Type~{II}. Below we shall see that $\underline{n}[H>\mu]=0 $. 

We state the following theorems for real $\lambda>0$. 
Using arguments based on results on the analyticity of Laplace transforms it can be
shown that they extend to a larger complex domain.
\begin{lemma}[On the joint law of $R$ and $H$ under $\protect\underline{n}$]
\label{LapRH} Suppose $\lambda>0$ and $0<y\leq\mu$, then we have 
\begin{equation}\label{fhfh}
\underline{n}[1-e^{-\lambda R};H<y]= -\frac1{2y}+\frac12\sqrt{2\lambda}\coth(y\sqrt{2\lambda}).
\end{equation}
\end{lemma}

\begin{proof}
For this proof we use the description of the trading excursion measure given
in Section~\ref{x-measure}.
From (\ref{Zx}) we note first that $R(Z^x)=T_x(\rho)+T_x(\tilde\rho)$ 
and $H(Z^x)=x$ for $0<x<\mu$.
The random variables $T_x(\rho)$ and $T_x(\tilde\rho)$ are independent
first hitting times of $\BES^3$-processes for level $x$.
The corresponding Laplace transform is well-known, namely
\begin{equation}
E[e^{-\lambda T_x(\rho)}]=
E[e^{-\lambda T_x(\tilde\rho)}]=
\frac{x\sqrt{2\lambda}}{\sinh(x\sqrt{2\lambda})}.
\end{equation}
See \cite[(3.8), p.762]{Ken1978} with $\nu=1/2$,
see also \cite[Tab.2, Row 3, Col.1, p.450 and Sec.4.5, p.453]{BPY2001}.
From Corollary~\ref{CorWilT} we get
\begin{eqnarray}
&&\nbid[1-e^{-\lambda T};H<y]
=
\frac12\int_0^\mu E[1-e^{-\lambda R(Z^x)};H(Z^x)<y]x^{-2}dx\\
&&\qquad=
\frac12\int_0^y E[1-e^{-\lambda R(Z^x)}]x^{-2}dx
=
\frac12\int_0^y
\left(
1-E\left[e^{-\lambda (T_x(\rho)+T_x(\tilde\rho))}\right]
\right)x^{-2}dx\\
&&\qquad=
\frac12\int_0^y\left[
1-\left(\frac{x\sqrt{2\lambda}}{\sinh(x\sqrt{2\lambda})}\right)^2
\right]x^{-2}dx
=
-\frac1{2y}+\frac12\sqrt{2\lambda}\coth(y\sqrt{2\lambda}).
\end{eqnarray}
\end{proof}

\begin{theorem}[Hyperbolic table under the trading excursion measure]
Let $\lambda>0$.\label{nhyp}

\begin{enumerate}
\item \label{nTI} We have for the length $R_{I}$ of the trading excursion to
the next Type $I$ trade
\begin{equation}
\underline{n}[1-e^{-\lambda R_{I}}]=-\frac{1}{2\mu }+\frac{1}{2}\sqrt{2\lambda }\coth (\mu \sqrt{2\lambda }),  \label{coth}
\end{equation}

\item \label{nTII} for the length $R_{II}$ of the trading excursion to the
next Type $II$ trade 
\begin{equation}
\underline{n}[1-e^{-\lambda R_{II}}]=\frac{1}{2\mu }-\sqrt{2\lambda }\csch(2\mu \sqrt{2\lambda }),  \label{csch}
\end{equation}

\item \label{nT} and for the length $R$ of the trading excursion to the next
trade 
\begin{equation}
\underline{n}[1-e^{-\lambda R}]=\frac{1}{2}\sqrt{2\lambda }\tanh (\mu \sqrt{2\lambda}).  \label{tanh}
\end{equation}
\end{enumerate}
\end{theorem}

\begin{proof}
Part~(\ref{nTI}) is due to \cite[p.66]{YY} and agrees with the result from Theorem~\ref{LapRH} for the special case $y=\mu$.

For Part~(\ref{nTII}) we note first from (\ref{Zmu}) that $R(Z^\mu)=T_\mu(\rho)+T_\mu(|b|)$,
with $\rho$ a $\BES^3$-process and $|b|$ an independent reflected Brownian motion,
The random variables $T_\mu(\rho)$ and $T_\mu(|b|)$ are their
first hitting times of level $\mu$ respectively, and $H(Z^\mu)=\mu$.
The corresponding Laplace transforms are well-known, namely
\begin{equation}
E[e^{-\lambda T_\mu(\rho)}]=\mu\sqrt{2\lambda}\csch(\mu\sqrt{2\lambda}),\quad
E[e^{-\lambda T_\mu(|b|)}]=\frac1{\cosh(\mu\sqrt{2\lambda})}.
\end{equation}
See \cite[(3.8), p.762]{Ken1978} with $\nu=1/2$ and $\nu=-1/2$, 
see also \cite[Tab.2, Row 3, Col.3, p.450 and Sec.4.5, p.453]{BPY2001}.
Thus we get by Corollary~\ref{CorWilT}
\begin{eqnarray}
\nbid[e^{-\lambda R_{II}};H=\mu]&=&
\frac1{2\mu}E[e^{-\lambda(T_\mu(\rho)+T_\mu(|b|))}]\\
&=&
\frac{\sqrt{2\lambda}}{2\sinh(\mu\sqrt{2\lambda})}
\frac1{\cosh(\mu\sqrt{2\lambda})}
=\sqrt{2\lambda}\csch(2\mu\sqrt{2\lambda}).
\end{eqnarray}

Part~(\ref{nT}) is obtained by adding Parts~(\ref{nTI}) and~(\ref{nTII}) and an elementary
duplication formula for hyperbolic functions, \cite[4.5.31, p.84]{AS}.
\end{proof}

\begin{remark}
We can rewrite (\ref{coth}) and (\ref{csch}) as follows: 
\begin{equation}  \label{coth-alt}
\underline{n}[1-e^{-\lambda R}I_{H<\mu}]= \frac12\sqrt{2\lambda}\coth(\mu\sqrt{2\lambda}),
\end{equation}
and 
\begin{equation}  \label{csch-alt}
\underline{n}[e^{-\lambda R}I_{H=\mu}]=\sqrt{2\lambda}\csch(2\mu\sqrt{2\lambda}).
\end{equation}
Note that $\underline{n}[e^{-\lambda R}I_{H<\mu}]=\infty$ for all $\lambda>0 $ though.
\end{remark}

We can describe the distributions of $R$ and $H$ under $\underline{n}$ more
explicitly using theta functions. There are many notations and
parametrizations for theta functions, see \cite[Sec.21.9, p.487]{WW} for an
overview. We choose a variant inspired by \cite{Dev2009}, which allows a
simple statement of transformation formulas. Let\footnote{With this system of notation we would have $\theta _{1}(x)\equiv 0$, thus it
is not mentioned here.} 
\begin{eqnarray}
\theta _{2}(x) &=&2\sum_{n\geq 1}e^{-(n-1/2)^{2}\pi x}, \\
\theta _{3}(x) &=&1+2\sum_{n\geq 1}e^{-n^{2}\pi x},\label{t3} \\
\theta _{4}(x) &=&1+2\sum_{n\geq 1}(-1)^{n}e^{-n^{2}\pi x}.
\end{eqnarray}

\begin{theorem}[Theta table]\label{thetatable}
We have
\begin{enumerate}
\item 
\begin{equation}\label{f1}
\underline{n}[R>x,H<y]=\frac{1}{2y}\left[\theta _{3}\left( \frac{\pi x}{2y^{2}}\right)-1\right],\quad x>0,0<y\leq\mu
\end{equation}

\item 
\begin{equation}\label{f2}
\underline{n}[R>x,H=\mu ]=\frac{1}{2\mu }\left[ 1-\theta _{4}\left( \frac{\pi x}{8\mu ^{2}}\right) \right] ,\quad x>0,
\end{equation}

\item 
\begin{equation}\label{f3}
\underline{n}[R>x]=\frac{1}{2\mu }\theta _{2}\left( \frac{\pi x}{2\mu ^{2}}\right) ,\quad x>0.
\end{equation}
\end{enumerate}
\end{theorem}
\begin{proof}
By Fubini's Theorem and (\ref{fhfh}) we compute the Laplace transform
\begin{eqnarray}
\int_0^\infty e^{-\lambda x}\underline{n}[R>x,H<y]dx&=&\frac1\lambda\underline{n}[1-e^{-\lambda R};H<y]\\
&=&-\frac1{2\lambda y}+\frac1{\sqrt{2\lambda}}\coth(y\sqrt{2\lambda}).
\end{eqnarray}
This agrees with the Laplace transform of (\ref{f1}), which is known resp.\ easily checked
by termwise-transformation followed by an application of the partial fraction expansion of the hyperbolic
cotangent.
Equations (\ref{f2}) and (\ref{f3}) can be proved in a similar way.
\end{proof}
For later usage we differentiate those formulas with respect to $x$ and $y$ and obtain 
\begin{equation}
\underline{n}[R\in dx]=-\frac{\pi }{4\mu ^{3}}\theta _{2}^{\prime }\left( 
\frac{\pi x}{2\mu ^{2}}\right) dx.
\end{equation}
\begin{eqnarray}
\underline{n}[R \in dx,H\in dy]&=& 
\left[ \frac{3\pi }{4y^4}\theta _{3}^{\prime }\left( \frac{\pi x}{2y^{2}}
\right) +\frac{\pi ^{2}x}{4y^{6}}\theta _{3}^{\prime \prime }\left( \frac{\pi x}{2y^{2}}\right) \right] I_{(0,\mu )}(y)dxdy\\
&&\qquad\qquad-\frac{\pi }{16\mu ^{3}}\theta_{4}^{\prime }\left( \frac{\pi x}{8\mu ^{2}}\right) dx\delta _{\mu
}(dy).
\end{eqnarray}

\subsection{The hyperbolic table under the probability measure}

We have two general devices for Poisson point processes to relate results
for the probability measure to results about its intensity measure, namely
the \emph{Exponential Formula} \cite[Prop.XII.1.12, p.476]{RY} and 
the \emph{Master Formula} \cite[Prop.XII.1.10 and Corl.XII.1.11, p.475]{RY}. 

\begin{theorem}[Hyperbolic table in exponential form]\label{hypexp}
Let $\lambda>0$ and $t>0$.

\begin{enumerate}
\item We have for the length $R_{I}$ of the trading excursion to the next
Type $I$ trade 
\begin{equation}
E\left[ \exp \left\{ -\lambda \sum_{0<s\leq t}R_{I}(\underline{e}_{s})\right\} \right] =\exp \left[ -t\left( -\frac{1}{2\mu }+\frac{1}{2}\sqrt{2\lambda }\coth (\mu \sqrt{2\lambda })\right) \right] ,
\end{equation}

\item for the length $R_{II}$ of the trading excursion to the next Type $II$
trade 
\begin{equation}
E\left[ \exp \left\{ -\lambda \sum_{0<s\leq t}R_{II}(\underline{e}_{s})\right\} \right] =\exp \left[ -t\left( \frac{1}{2\mu }-\sqrt{2\lambda}\csch(2\mu \sqrt{2\lambda })\right) \right] ,
\end{equation}

\item and for the length $R$ of the trading excursion to the next trade 
\begin{equation}
E\left[ \exp \left\{ -\lambda \sum_{0<s\leq t}R(\underline{e}_{s})\right\} \right] =\exp \left[ -t\left( \frac{1}{2}\sqrt{2\lambda }\tanh (\mu \sqrt{2\lambda })\right) \right] .
\end{equation}
\end{enumerate}
\end{theorem}

\begin{proof}
This follows from the exponential formula with $f(s,u)=\lambda R_{I}(u)$, $f(s,u)=\lambda R_{II}(u)$, $f(s,u)=\lambda R(u)$
and Theorem~\ref{nhyp} above.
\end{proof}

\begin{remark}
The sums on the left hand sides are summing over excursions until the local time reaches the level $t$, which corresponds to real time $\tau _{t}$.
\end{remark}

\begin{corollary}[Hyperbolic table in additive form]
Suppose $\lambda>0$ and $t>0$.

\begin{enumerate}
\item We have for the length $R_{I}$ of the trading excursion to the next
Type $I$ trade 
\begin{equation}
E\left[ \sum_{0<s\leq t}\left( 1-e^{-\lambda R_{I}(\underline{e}_{s})}\right) \right] =t\left[ -\frac{1}{2\mu }+\frac{1}{2}\sqrt{2\lambda }\coth (\mu \sqrt{2\lambda })\right].
\end{equation}

\item We have for the time to the next trade $T$ assuming it is Type~{II} 
\begin{equation}
E\left[ \sum_{0<s\leq t}\left( 1-e^{-\lambda R_{II}(\underline{e}
_{s})}\right) \right] =t\left[ \frac{1}{2\mu }-\sqrt{2\lambda }\csch(2\mu 
\sqrt{2\lambda })\right] .
\end{equation}

\item We have for the time to the next trade $T$ 
\begin{equation}
E\left[ \sum_{0<s\leq t}\left( 1-e^{-\lambda R(\underline{e}_{s})}\right) 
\right] =t\left[ \frac{1}{2}\sqrt{2\lambda }\tanh (\mu \sqrt{2\lambda })
\right] .
\end{equation}
\end{enumerate}
\end{corollary}

\begin{proof}
This follows from the master formula \cite[XII.1.10]{RY} with 
$f(s,u)=1-e^{-\lambda R_{I}(u)}$, 
$f(s,u)=1-e^{-\lambda R_{II}(u)}$, 
$f(s,u)=1-e^{-\lambda R(u)}$ 
respectively and Theorem~\ref{nhyp} above.
\end{proof}

\section{Laplace transform for the avalanche length}

\label{subsec:notwo}

Orders in the LOB get executed via avalanches. In other words, limit orders
may accumulate on some levels, and when the price process crosses those
values, we will see a sudden decrease of the number of orders. We take
record if there is no order execution in a time period lasting longer than $\varepsilon >0$.

Recall from Definitions \ref{trading time points} and \ref{d:first and last trading time} that $\Theta$ denotes the set of all trading times, $\Upsilon(t)$ (resp.\ $\Xi$) denotes the time of last trade before (resp.\ next trade after) time $t$.

\begin{definition}
Let
\begin{equation}
a\in \Theta ,~\Upsilon (a)\leq (a-\varepsilon )_{+},~b\in \Theta ,~\Xi
(b)\geq b+\varepsilon ,
\end{equation}
\begin{equation}
\Xi (t)\leq t+\varepsilon \quad \forall t\in (a,b).
\end{equation}
An $\varepsilon $-\emph{avalanche }is defined as the process $\left\{
W_{t}:a\leq t\leq b\right\} $. We call $a$ and $b$ start and end of the
avalanche. The corresponding $\varepsilon $\emph{-avalanche length} is $b-a$.
\end{definition}

There is a sequence of stopping times $(T_{n}^{a})_{n\geq 1}$ enumerating
the start of avalanches, and a sequence of honest times $(T_{n}^{e})_{n\geq
1}$ enumerating the end of avalanches (for the completed filtration). We are interested into the distribution of the avalanche length for which we
will rely on the hyperbolic table of the distribution of intertrading times.
\begin{theorem}
Let $T$ be a stopping time starting an $\varepsilon$-avalanche and $A^\varepsilon$ be the corresponding avalanche length. Then we have the
Laplace transform 
\begin{equation}
E\left[ e^{-\lambda A^{\varepsilon}}\right]= \frac{H(\varepsilon)}{H(\varepsilon)+\int_0^\varepsilon(1-e^{-\lambda x})h(x)dx}.
\label{L-trafo full avalanche length}
\end{equation}
where 
\begin{equation}
H(\varepsilon)=\frac1{2\mu}\theta_2\left(\frac{\pi\varepsilon}{2\mu^2}\right), \quad h(x)=-\frac{\pi}{4\mu^3}\theta_2^{\prime }\left(\frac{\pi x}{2\mu^2}\right).
\end{equation}
\end{theorem}

\begin{proof}
Let $\ebid$ denote the trading excursion process for the ask side.
We have seen that it is a Poisson point process with intensity measure
$\nbid$. Let $R$ denote the excursion length functional.
Set
\begin{equation}
X_s=\sum_{0\leq r\leq s}R(e_r),\quad s\geq0.
\end{equation}
By Theorem~\ref{hypexp} and Theorem~\ref{thetatable} we see that $X$ is a L\'evy process with L\'evy measure
\begin{equation}
\nu_X(dx)=\nbid(dx) = h(x)dx,\quad x>0.
\end{equation}
For $\varepsilon>0$ we can write $X=J+Y$ with
\begin{equation}
Y_s=\sum_{0\leq r\leq s}\Delta X_sI_{\Delta X_s>\varepsilon},\quad
J_s=X_s-Y_s,\quad s\geq0.
\end{equation}
Then $J$ and $Y$ are two independent L\'evy processes with L\'evy measures
\begin{equation}
\nu_J(dx)=I_{x\leq\varepsilon}h(x)dx,\quad
\nu_Y(dx)=I_{x>\varepsilon}h(x)dx,
\quad x>0.
\end{equation}
Let $S=\inf\{s\geq0:\Delta Y_s>0\}$.
This is the first jump time of a compound Poisson process with L\'evy measure
$\nu_Y$ and thus exponential with parameter $\beta$ given by
\begin{equation}
  \beta=\nu_Y(\mathbb R_+)=\nbid[R>\epsilon]=H(\epsilon).
\end{equation}
The L\'evy-Khintchine formula for the cumulant of $J$ says
\begin{equation}
\kappa(\lambda)=\int_0^\infty(e^{-\lambda x}-1)\nu_J(dx).
\end{equation}
A straight integration gives
\begin{equation}
\kappa(\lambda)=-\int_0^\varepsilon(1-e^{-\lambda x})h(x)dx.
\end{equation}
The full avalanche length is $A=J_S$. By independence we obtain the Laplace transform
\begin{equation}
E[e^{-\lambda A}]=
\int_0^\infty e^{\kappa(\lambda)s}\beta e^{-\beta s}ds=\frac{\beta}{\beta-\kappa(\lambda)}
\end{equation}
and combining this with the results above yields the result.
\end{proof}

\begin{remark}
Let us ignore Type~{II} trades, and assume that orders are only executed as
in the Type~{I} case. Dassios and Wu \cite{DW} derive the Laplace transform
of the avalanche length $L^{\varepsilon }$ in the context of Parisian
options. The same formula can be inferred (Dudok de Wit \cite{DDW}) from the
L\'{e}vy measure of the subordinator consisting of Brownian passage times.
It results that
\begin{equation}
E\left[ e^{-\lambda L^{\varepsilon }}\right] =\frac{1}{\sqrt{\lambda
\varepsilon \pi }\mathrm{\func{erf}}\left( \sqrt{\lambda \varepsilon }\right) +e^{-\lambda \varepsilon }},  \label{Laurent}
\end{equation}
which can be proven in a completely analogous way by 
choosing $\tilde{h}\left( x\right) =x^{-3/2}/\sqrt{2\pi }$, which is the density
of the excursion length $T$ under the Ito measure $n$, instead of function $h$. In this
case the integral in the denominator can be evaluated in terms of the error
function by some elementary computations.
\end{remark}


\section{Technical remarks, discussions and proofs}

\subsection{Depth of the order book after the first trade}

The next lemma states that the limit order book has an order depth of at
least $\mu$ after the first trade $\tau_1$ has happened. Here, we work under
the same assumptions as in Section \ref{s:tradingtimes}.

\begin{lemma}
\label{l:depths of lob} $\tau_1$ is the first trade and for any $t\geq
\tau_1 $ there is a closed set $K_t$ which is a Lebesgue null-set such that 
\begin{equation*}
(\alpha(t),\alpha(t)+\mu)\backslash K_t \subseteq \{x\in\mathbb{R}:
V(t,x)>0\} \subseteq (\alpha(t),\infty).
\end{equation*}
\end{lemma}

\begin{proof}
  The last inclusion is trivial.

  For $t\in[0,\tau_1]$ we have $V(t,x) = 0$ for any $x<w_\ast(0,\tau_1)+\mu$ and, hence, the first trade does not take place in $[0,\tau_1)$. 
  Moreover, we have 
  $$\{w(t):t\in[0,\tau_1]\} = [w_\ast(0,\tau_1),w^*(0,\tau_1)] = [w_\ast(0,\tau_1),w_\ast(0,\tau_1)+\mu].$$
  Consequently, $L_{\tau_1}^{\cdot}$ has support $[w_\ast(0,\tau_1),w_\ast(0,\tau_1)+\mu]$. Let 
  $$K_{\tau_1}:=\{x\in[w_\ast(0,\tau_1),w_\ast(0,\tau_1)+\mu]: L_{\tau_1}^x=0\} = \{x\in[w^\ast(0,\tau_1),w^\ast(0,\tau_1)+\mu]: L_{\tau_1}^{x-\mu}=0\}.$$
Then, $K_{\tau_1}$ is closed by continuity of the occupation density and it is a Lebesgue null-set. Clearly, we have 
  $$ \{x: V(\tau_1,x)>0 \} = \{x: L_{\tau_1}^{x-\mu}>0\} = [w^\ast(0,\tau_1),w^\ast(0,\tau_1)+\mu]\backslash K_{\tau_1}. $$
  The first claim follows.
  
 Let $T \in [\tau_1,\infty]$ be maximal such that for any $t\in [\tau_1,T)$ the claim holds. Assume by contradiction that $T<\infty$. By continuity of $w$ the claim holds at time $T$. Again by continuity there is $\delta>0$ such that $w^\ast(T,T+\delta) < w_\ast(T,T+\delta)$. We have
  $$ V(t,x) = V(T,x)\mathbf1_{\sigma_{T,t}(x)=T} + L^{x-\mu}_t-L^{x-\mu}_{\sigma_{T,t}(x)} $$
 where $\sigma_{T,t}(x) := \inf\{s\in [T,t]: w(s) = x\text{ or }s=T\}$ for any $x\in\mathbb R$, $t\in [T,T+\delta]$. For $x> w^*(T,t)$ we have $\sigma_{T,t}(x) = T$ and hence 
 $$ V(t,x) = V(T,x) + L^{x-\mu}_t-L^{x-\mu}_{T} $$
 and for $x \leq w^*(T,t)$ we have $ V(t,x) = 0$. Thus,
  $$ \{x: V(t,x) >0\} = \{x > w^*(T,t) : V(T,x)>0\} \cup \{x \in\mathbb R: L^{x-\mu}_t-L^{x-\mu}_{T}>0 \}. $$
This clearly contradicts the maximality of $T$. Thus $T=\infty$. The second claim follows.
\end{proof}

\subsection{Proper trades\label{proper trades}}

The condition for $\tau$ to be a trading time of the path $w$ means that the
order book is not void in any sufficiently small interval $(w(\tau),w\left(
\tau\right) +\varepsilon)$. This does not necessarily imply that an actual
trade takes place at $\tau$. In fact, if the order book is initially empty,
then 
\begin{equation*}
\tau:=\inf\left\{ t>0:w(t)-\inf\{w(s):s\in (0,t)\}=\mu\right\}
\end{equation*}
is the time of the first trade and, if additionally $L_\tau^{\inf\{w(s):s\in
(0,t)\}} = 0$ (which happens if the occupation density is continuous), then $V(\tau-,\alpha(\tau)) = 0$, i.e.\ the limit order book has orders only right
above the level $\alpha(\tau)$. However, this phenomenon is somewhat an
artifact of working in continuous time, and for the purposes of the current
study it is quite sensible to include such times as well under the label
`trading times'. In fact, a proper trade takes place at time $\tau>0$ iff $V\left( \tau {-},w({\tau})\right) >0$ (and then, as always, $V\left( \tau,w({\tau })\right) =0$).

Note, whenever the order volume can be described by a continuous function
then the volume at the best bid and ask will be zero.

Finally, we want to identify the proper trades. First, we need to know that the limit order book has at
least an order depth of $\mu$ starting from the first trading time $\tau_1$.

\begin{lemma}
\label{l:lob open interval} Let $\eta$ be a random time with $\eta\geq
\tau_1 $ where $\tau_1$ is defined in Definition~\ref{type II trades}. Then
we have 
\begin{equation*}
(\alpha(\eta),\alpha(\eta)+\mu)\subseteq \{x\in\mathbb{R}:V(\eta,x)>0\}~\mbox{$P$-a.s.}
\end{equation*}
\end{lemma}

\begin{proof}
  We have 
    \begin{align*}
     \tau_1 &= \inf\{t>0: W_\ast(0,t)+\mu = W^\ast(0,t) \}, \\
     V(t,x) &= L_t^{x-\mu} - L^{x-\mu}_{\sigma_t(x)}, \\
     \{x\in\mathbb R:L_t^{x-\mu}>0\} &= (W_\ast(0,t)+\mu,W^\ast(0,t)+\mu),\quad P\text{-a.s.}.
    \end{align*}  
 Thus, we get 
 $$\{x\in\mathbb R: V(\tau_1,x) > 0\} = \{W^\ast(0,\tau_1),W^\ast(0,\tau_1)+\mu\}=(\alpha(\tau_1),\alpha(\tau_1)+\mu)
 \quad\mbox{$P$-a.s.}$$
Thus, for a random time $\eta\geq \tau_1$ we have
$$\{x\in\mathbb R: V(\eta,x) > 0\} \supseteq (\alpha(\eta),\alpha(\eta)+\mu)\quad\mbox{$P$-a.s.}$$
\end{proof}
The next proposition identifies the proper trades exactly as the trades of
Type~{Ia} and Type~{Ib}.
\begin{proposition}
\label{p: proper trades} Let $\Theta_!:=\{t>0: V(t-,w(t))>0\}$. Then 
\begin{equation*}
P(\Theta_! = \Theta_{\mathrm{I}_a} \cup \Theta_{\mathrm{I}_b}) = 1.
\end{equation*}
In other words, a proper trade takes place if and only if a trade of Type~{Ia} or Type~{Ib}
takes place.
\end{proposition}

\begin{proof}
 Let $t\in\Theta_!$. Let $\delta>0$ such that $V(t-\epsilon,w(t))>0$ for any $\epsilon\in(0,\delta)$. 
 Now let $\epsilon\in(0,\delta)$. Lemma \ref{l:lob open interval} yields that $\alpha(t-\epsilon) < w(t)$. 
 Thus, there is $t_\epsilon\in (t-\epsilon,t)$ such that $\alpha(t_\epsilon)=w(t_\epsilon)$. Hence, 
 we have $t_\epsilon\in\Theta$. Consequently, we have $\lasttrade(t) = t$ which 
 implies $t\in \Theta_{\mathrm{I}_a} \cup \Theta_{\mathrm{I}_b}$.
 
 Now, let $t\in \Theta_{\mathrm{I}_a} \cup \Theta_{\mathrm{I}_b}$. 
 Proposition~\ref{structure alpha} yields that there is $t_0\in[0,t]$ such 
 that $\alpha(s)=w^*(t_0,s)$ for any $s\in[t_0,t]$. Since $t\notin\Theta_{\mathrm{II}}$ we 
 have $t_0<t$. Hence, there is $t_1\in[t_0,t)$ such that $\alpha(t_1)+\mu > \alpha(t)$. 
 Lemma~\ref{l:lob open interval} yields that $V(t_1,x) > 0$ for 
 any $x\in (\alpha(t_1),\alpha(t_1)+\mu)\supseteq [\alpha(t),\alpha(t_1)+\mu)$. 
 Consequently, we have $V(s,x)>0$ for any $s\in [t_1,t)$, $x\in [\alpha(t),\alpha(t_1)+\mu)$. 
 This implies that $V(t-,w(t)) = V(t-,\alpha(t)) > 0$. Hence, $t\in\Theta_!$.
\end{proof}

\newcommand{\etalchar}[1]{$^{#1}$}


\begin{thebibliography}{TLD{\etalchar{+}}11}

\bibitem[AJ13]{AJ}
Fr{\'e}d{\'e}ric Abergel and Aymen Jedidi.
\newblock A mathematical approach to order book modeling.
\newblock {\em International Journal of Theoretical and Applied Finance},
  16(5):1350025 (40 pages), 2013.

\bibitem[AS64]{AS}
Milton Abramowitz and Irene~A. Stegun.
\newblock {\em Handbook of Mathematical Functions with Formulas, Graphs, and
  Mathematical Tables}, volume~55 of {\em National Bureau of Standards Applied
  Mathematics Series}.
\newblock For sale by the Superintendent of Documents, U.S. Government Printing
  Office, Washington, D.C., 1964.

\bibitem[BHQ14]{BHQ}
Christian Bayer, Ulrich Horst, and Jinniao Qiu.
\newblock A functional limit theorem for limit order books with state dependent
  price dynamics.
\newblock {\em ArXiv e-prints}, May 2014.

\bibitem[Blu92]{Blu}
Robert~M. Blumenthal.
\newblock {\em Excursions of Markov Processes}.
\newblock Birkh\"auser Boston Inc., Boston, MA, 1992.

\bibitem[BPY01]{BPY2001}
Philippe Biane, Jim Pitman, and Marc Yor.
\newblock Probability laws related to the {J}acobi theta and {R}iemann zeta
  functions, and {B}rownian excursions.
\newblock {\em Bulletin of the American Mathematical Society}, 38(4):435--465,
  2001.

\bibitem[CdL13]{CL}
Rama Cont and Adrien de~Larrard.
\newblock Price dynamics in a {M}arkovian limit order market.
\newblock {\em SIAM Journal on Financial Mathematics}, 4(1):1--25, 2013.

\bibitem[CJP15]{Cartea}
\'Alvaro Cartea, Sebastian Jaimungal, and Jose Penalva.
\newblock {\em Algorithmic and High-Frequency Trading}.
\newblock Cambridge University Press, 2015.

\bibitem[CST10]{CST}
Rama Cont, Sasha Stoikov, and Rishi Talreja.
\newblock A stochastic model for order book dynamics.
\newblock {\em Operations Research}, 58(3):549--563, 2010.

\bibitem[DdW13]{DDW}
Laurent Dudok~de Wit.
\newblock {\em Liquidity Risks Based on the Limit Order Book}.
\newblock Master thesis, TU Wien and EPFL, 2013.

\bibitem[Dev09]{Dev2009}
Luc Devroye.
\newblock On exact simulation algorithms for some distributions related to
  {J}acobi theta functions.
\newblock {\em Statistics \& Probability Letters}, 79(21):2251--2259, 2009.

\bibitem[DRR13]{dLRR}
Sylvain Delattre, Christian~Y. Robert, and Mathieu Rosenbaum.
\newblock Estimating the efficient price from the order flow: a {B}rownian
  {C}ox process approach.
\newblock {\em Stochastic Processes and their Applications}, 123(7):2603--2619,
  2013.

\bibitem[DW15]{DW}
Angelos Dassios and Shanle Wu.
\newblock Two-side {Parisian} option with single barrier.
\newblock Preprint, LSE, Statistics, 2015.

\bibitem[EK86]{ethier.kurtz.86}
Stewart~N. Ethier and Thomas~G. Kurtz.
\newblock {\em Markov Processes}.
\newblock John Wiley \& Sons, Inc., New York, 1986.
\newblock Characterization and Convergence.

\bibitem[Gau02]{G}
Laurent Gauthier.
\newblock {\em Options R{\'e}elles et Options Exotiques, une Approche
  Probabiliste}.
\newblock Th{\`e}se pour le doctorat, Universit{\'e} Panth{\'e}on-Sorbonne,
  Paris~I, November 2002.

\bibitem[GH80]{geman.horowitz.80}
Donald Geman and Joseph Horowitz.
\newblock Occupation densities.
\newblock {\em The Annals of Probability}, 8(1):1--67, 1980.

\bibitem[HKR17a]{HKR-DISCRETE}
Friedrich Hubalek, Paul Kr\"{u}hner, and Thorsten Rheinl\"{a}nder.
\newblock The discrete {LOB} draft.
\newblock Work in progress, {TU Wien}, 2017.

\bibitem[HKR17b]{HKR-SPDE}
Friedrich Hubalek, Paul Kr\"{u}hner, and Thorsten Rheinl\"{a}nder.
\newblock The {SPDE}-{LOB} draft.
\newblock Work in progress, {TU Wien}, 2017.

\bibitem[Ken78]{Ken1978}
John Kent.
\newblock Some probabilistic properties of {B}essel functions.
\newblock {\em The Annals of Probability}, 6(5):760--770, 1978.

\bibitem[Kru03]{K}
{\L}ukasz Kruk.
\newblock Functional limit theorems for a simple auction.
\newblock {\em Mathematics of Operations Research}, 28(4):716--751, 2003.

\bibitem[KS91]{KS}
Ioannis Karatzas and Steven~E. Shreve.
\newblock {\em Brownian Motion and Stochastic Calculus}.
\newblock Springer-Verlag, New York, second edition, 1991.

\bibitem[Ost06]{O}
J{\"o}rg Osterrieder.
\newblock {\em Arbitrage, Market Microstructure and the Limit Order Book}.
\newblock {PhD} thesis, ETH Z\"urich, 2006.

\bibitem[Pro04]{protter.04}
Philip~E. Protter.
\newblock {\em Stochastic Integration and Differential Equations}.
\newblock Springer-Verlag, Berlin, second edition, 2004.

\bibitem[Rog81]{Rog}
L.~C.~G. Rogers.
\newblock Williams' characterisation of the {B}rownian excursion law: proof and
  applications.
\newblock In {\em Seminar on {P}robability, {XV} ({U}niv. {S}trasbourg,
  {S}trasbourg, 1979/1980)}, volume 850 of {\em Lecture Notes in Math.}, pages
  227--250. Springer, Berlin-New York, 1981.

\bibitem[RY99]{RY}
Daniel Revuz and Marc Yor.
\newblock {\em Continuous martingales and {B}rownian motion}.
\newblock Springer-Verlag, Berlin, third edition, 1999.

\bibitem[SC06]{SC}
Matthew Stapleton and Kim Christensen.
\newblock One-dimensional directed sandpile models and the area under a
  {B}rownian curve.
\newblock {\em Journal of Physics. A. Mathematical and General},
  39(29):9107--9126, 2006.

\bibitem[TLD{\etalchar{+}}11]{TothEtAl2011}
B.~T\'oth, Y.~Lemp\'eri\`ere, C.~Derembleand, J.~De~Lataillade, J.~Kockelkoren,
  and J.-P. Bouchaud.
\newblock Anomalous price impact and the critical nature of liquidity in
  financial markets.
\newblock {\em Phys. Rev. X}, 1(2):021006, 2011.

\bibitem[Wil91]{PWM}
David Williams.
\newblock {\em Probability with Martingales}.
\newblock Cambridge University Press, Cambridge, 1991.

\bibitem[WW96]{WW}
Edmund~T. Whittaker and George~N. Watson.
\newblock {\em A Course of Modern Analysis}.
\newblock Cambridge University Press, Cambridge, 1996.
\newblock Reprint of the fourth (1927) edition.

\bibitem[YY13]{YY}
Ju-Yi Yen and Marc Yor.
\newblock {\em Local Times and Excursion Theory for {B}rownian Motion}, volume
  2088 of {\em Lecture Notes in Mathematics}.
\newblock Springer, 2013.
\newblock A tale of Wiener and It{\^o} measures.

\end{thebibliography}
\end{document}